\documentclass[pra,aps,showpacs,twocolumn,superscriptaddress]{revtex4}

\usepackage{amsmath,amsfonts,amssymb,caption,hyperref,color,epsfig,graphics,subfigure,graphicx,latexsym,mathrsfs,revsymb,theorem,url,verbatim,epstopdf}
\allowdisplaybreaks[3]

\hypersetup{colorlinks,linkcolor={blue},citecolor={blue},urlcolor={red}}

\usepackage{hyperref}

\makeatletter

\newcommand{\Rmnum}[1]{\expandafter\@slowromancap\romannumeral #1@}
\makeatother
\newtheorem{definition}{Definition}
\newtheorem{proposition}[definition]{Proposition}
\newtheorem{Lemma}[definition]{Lemma}

\newtheorem{Theorem}[definition]{Theorem}

\newtheorem{conjecture}[definition]{Conjecture}

\newtheorem{remark}[definition]{Remark}

\newtheorem{example}{Example}
\newtheorem{question}[definition]{Question}

\def\squareforqed{\hbox{\rlap{$\sqcap$}$\sqcup$}}
\def\qed{\ifmmode\squareforqed\else{\unskip\nobreak\hfil
		\penalty50\hskip1em\null\nobreak\hfil\squareforqed
		\parfillskip=0pt\finalhyphendemerits=0\endgraf}\fi}
\def\endenv{\ifmmode\;\else{\unskip\nobreak\hfil
		\penalty50\hskip1em\null\nobreak\hfil\;
		\parfillskip=0pt\finalhyphendemerits=0\endgraf}\fi}
\newenvironment{proof}{\noindent \textbf{{Proof.~} }}{\qed}
\def\Dbar{\leavevmode\lower.6ex\hbox to 0pt
	{\hskip-.23ex\accent"16\hss}D}
\makeatletter
\def\url@leostyle{%
	\@ifundefined{selectfont}{\def\UrlFont{\sf}}{\def\UrlFont{\small\ttfamily}}}
\makeatother
\def\bcj{\begin{conjecture}}
	\def\ecj{\end{conjecture}}
\def\bcr{\begin{corollary}}
	\def\ecr{\end{corollary}}
\def\bd{\begin{definition}}
	\def\ed{\end{definition}}
\def\bea{\begin{eqnarray}}
\def\eea{\end{eqnarray}}
\def\bem{\begin{enumerate}}
	\def\eem{\end{enumerate}}
\def\bex{\begin{example}}
	\def\eex{\end{example}}
\def\bim{\begin{itemize}}
	\def\eim{\end{itemize}}
\def\bl{\begin{lemma}}
	\def\el{\end{lemma}}
\def\bma{\begin{bmatrix}}
	\def\ema{\end{bmatrix}}
\def\bpf{\begin{proof}}
	\def\epf{\end{proof}}
\def\bpp{\begin{proposition}}
	\def\epp{\end{proposition}}
\def\bqu{\begin{question}}
	\def\equ{\end{question}}
\def\br{\begin{remark}}
	\def\er{\end{remark}}
\def\bt{\begin{theorem}}
	\def\et{\end{theorem}}

\def\btb{\begin{tabular}}
	\def\etb{\end{tabular}}

\newcommand{\nc}{\newcommand}

\def\a{\alpha}
\def\b{\beta}

\def\r{\rho}
\def\s{\sigma}

\def\o{\omega}

\nc{\bbA}{\mathbb{A}} \nc{\bbB}{\mathbb{B}} \nc{\bbC}{\mathbb{C}}
\nc{\bbD}{\mathbb{D}} \nc{\bbE}{\mathbb{E}} \nc{\bbF}{\mathbb{F}}
\nc{\bbG}{\mathbb{G}} \nc{\bbH}{\mathbb{H}} \nc{\bbI}{\mathbb{I}}
\nc{\bbJ}{\mathbb{J}} \nc{\bbK}{\mathbb{K}} \nc{\bbL}{\mathbb{L}}
\nc{\bbM}{\mathbb{M}} \nc{\bbN}{\mathbb{N}} \nc{\bbO}{\mathbb{O}}
\nc{\bbP}{\mathbb{P}} \nc{\bbQ}{\mathbb{Q}} \nc{\bbR}{\mathbb{R}}
\nc{\bbS}{\mathbb{S}} \nc{\bbT}{\mathbb{T}} \nc{\bbU}{\mathbb{U}}
\nc{\bbV}{\mathbb{V}} \nc{\bbW}{\mathbb{W}} \nc{\bbX}{\mathbb{X}}
\nc{\bbZ}{\mathbb{Z}}

\nc{\bA}{{\bf A}} \nc{\bB}{{\bf B}} \nc{\bC}{{\bf C}}
\nc{\bD}{{\bf D}} \nc{\bE}{{\bf E}} \nc{\bF}{{\bf F}}
\nc{\bG}{{\bf G}} \nc{\bH}{{\bf H}} \nc{\bI}{{\bf I}}
\nc{\bJ}{{\bf J}} \nc{\bK}{{\bf K}} \nc{\bL}{{\bf L}}
\nc{\bM}{{\bf M}} \nc{\bN}{{\bf N}} \nc{\bO}{{\bf O}}
\nc{\bP}{{\bf P}} \nc{\bQ}{{\bf Q}} \nc{\bR}{{\bf R}}
\nc{\bS}{{\bf S}} \nc{\bT}{{\bf T}} \nc{\bU}{{\bf U}}
\nc{\bV}{{\bf V}} \nc{\bW}{{\bf W}} \nc{\bX}{{\bf X}}
\nc{\bZ}{{\bf Z}}

\nc{\cA}{{\cal A}} \nc{\cB}{{\cal B}} \nc{\cC}{{\cal C}}
\nc{\cD}{{\cal D}} \nc{\cE}{{\cal E}} \nc{\cF}{{\cal F}}
\nc{\cG}{{\cal G}} \nc{\cH}{{\cal H}} \nc{\cI}{{\cal I}}
\nc{\cJ}{{\cal J}} \nc{\cK}{{\cal K}} \nc{\cL}{{\cal L}}
\nc{\cM}{{\cal M}} \nc{\cN}{{\cal N}} \nc{\cO}{{\cal O}}
\nc{\cP}{{\cal P}} \nc{\cQ}{{\cal Q}} \nc{\cR}{{\cal R}}
\nc{\cS}{{\cal S}} \nc{\cT}{{\cal T}} \nc{\cU}{{\cal U}}
\nc{\cV}{{\cal V}} \nc{\cW}{{\cal W}} \nc{\cX}{{\cal X}}
\nc{\cZ}{{\cal Z}}

\nc{\hA}{{\hat{A}}} \nc{\hB}{{\hat{B}}} \nc{\hC}{{\hat{C}}}
\nc{\hD}{{\hat{D}}} \nc{\hE}{{\hat{E}}} \nc{\hF}{{\hat{F}}}
\nc{\hG}{{\hat{G}}} \nc{\hH}{{\hat{H}}} \nc{\hI}{{\hat{I}}}
\nc{\hJ}{{\hat{J}}} \nc{\hK}{{\hat{K}}} \nc{\hL}{{\hat{L}}}
\nc{\hM}{{\hat{M}}} \nc{\hN}{{\hat{N}}} \nc{\hO}{{\hat{O}}}
\nc{\hP}{{\hat{P}}} \nc{\hR}{{\hat{R}}} \nc{\hS}{{\hat{S}}}
\nc{\hT}{{\hat{T}}} \nc{\hU}{{\hat{U}}} \nc{\hV}{{\hat{V}}}
\nc{\hW}{{\hat{W}}} \nc{\hX}{{\hat{X}}} \nc{\hZ}{{\hat{Z}}}

\nc{\hn}{{\hat{n}}}

\def\max{\mathop{\rm max}}
\def\min{\mathop{\rm min}}

\def\tr{\mathop{\rm Tr}}

\newcommand{\bra}[1]{\langle#1|}
\newcommand{\ket}[1]{|#1\rangle}

\newcommand{\norm}[1]{\lVert#1\rVert}

\def\Dbar{\leavevmode\lower.6ex\hbox to 0pt
	{\hskip-.23ex\accent"16\hss}D}

\begin{document}
	\title{Families of Schmidt-number witnesses for high dimensional quantum states}
	
	\author{Xian Shi}\email[]
	{shixian01@gmail.com}
	\affiliation{College of Information Science and Technology,
		Beijing University of Chemical Technology, Beijing 100029, China}

	%
	
	
	
	\date{\today}
	\begin{abstract}
		Higher dimensional entangled states demonstrate significant advantages in quantum information processing tasks. Schmidt number is a quantity on the entanglement dimension of a bipartite state. Here we build families of $k$-positive maps from the symmetric information complete positive operator-valued measurements and mutually unbiased bases, and we also present the Schmidt number witnesses, correspondingly. At last, based on the witnesses obtained from mutually unbiased bases, we show the distance between a bipartite state and the set of states with Schmidt number less than $k$.
	\end{abstract}

	\pacs{03.65.Ud, 03.67.Mn}
	\maketitle

\section{Introduction}
\indent	Entanglement is one of the most fundamental features in quantum mechanics compared to classical physics \cite{horodecki2009quantum,plenio2014introduction}. It also plays critical roles in quantum information and quantum computation theory, such as quantum cryptography \cite{ekert1991quantum}, teleportation \cite{bennett1993teleporting}, and superdense coding \cite{bennett1992communication}. \par
From the start of quantum information theory, lots of efforts have been devoted to the problems of distinguishing whether a state is separable or entangled \cite{lewenstein2001characterization,chen2003matrix,rudolph2005further,guhne2006entanglement,zhang2008entanglement,spengler2012entanglement,shen2015separability,shang2018enhanced,sarbicki2020family,shi2023family,shi2024entanglement} and quantifying entanglement of the state \cite{wootters1998entanglement,christandl2004squashed,chen2005concurrence,de2007lower,li2020improved}. A commonly used method to certify the entanglement of a state is to build an effective entanglement witness. There are many ways to construct the entanglement witness \cite{guhne2009entanglement,chruscinski2014entanglement}. In 2018, Chruscinski $et$ $al.$ showed a method to construct entanglement witnesses from mutually unbiased bases (MUB) \cite{chruscinski2018entanglement}. Whereafter, the authors in \cite{li2019mutually,siudzinska2021entanglement,siudzinska2022indecomposability} generalized the method to build entanglement witnesses with other classes of positive operator-valued measurements (POVM).  One of the most essential entanglement measures is the Schmidt number (SN) \cite{terhal2000schmidt}, this quantity indicates the lowest dimension of the system needed 
to generate the entanglement. Furthermore, genuine high dimensional entanglement plays important roles in many quantum information tasks, such as, quantum communication \cite{cozzolino2019high}, quantum control \cite{kues2017chip} and universal quantum computation \cite{wang2020qudits,paesani2021scheme}. 

However, like most entanglement measures, it is hard to obtain the SN of a generic entangled state. Recently, the method to bound the SN of an entangled state attracted much attention from the relevant researchers \cite{bavaresco2018measurements,wyderka2023construction,liu2023characterizing,simon2023,liu2024bounding,tavakoli2024enhanced}. In \cite{bavaresco2018measurements}, Bavaresco $et$ $al.$ proposed a method to bound the dimension of an entangled state. Recently, Liu $et$ $al.$ presented the results of SN of a given state based on its covariance matrix \cite{liu2023characterizing,liu2024bounding}. Tavakoli and Morelli showed the bound of SN of a given state with the help of MUBs and SIC POVMs \cite{tavakoli2024enhanced}. Similar to entanglement, a straightford method to certify the dimension of an entangled state is by constructing the $k$-postive maps \cite{terhal2000schmidt} or SN witnesses \cite{terhal2000schmidt,sanpera2001schmidt,wyderka2023construction}. However, there are few results obtained on constructing the SN witnesses with the help of certain POVMs.

In this manuscript, we will present the methods to construct $k$-positive maps with the use of symmetric information complete (SIC) POVMs and MUBs, which generalizes the methods of \cite{chruscinski2018entanglement}. we also give the corresponding SN witness. Moreover, we compare the $k$-positive maps here with those in \cite{terhal2000schmidt,tomiyama1985geometry}. At last, we present the lower bounds of the distance between a bipartite mixed state and the set of states with Schmidt number less than $k$ based on the SN witnesses, which is built from MUBs. 
\section{Preliminary Knowledge}
In the manuscript, the quantum systems we considered here are finite dimensions. Next we denote $\mathcal{D}(\mathcal{H}_{AB})$ as the set consisting of the states of $\mathcal{H}_{AB},$  $$\mathcal{D}(\mathcal{H}_{AB})=\{\rho_{AB}|\rho_{AB}\ge0,\tr\rho_{AB}=1\}.$$ And we denote $\ket{\psi_d}=\sum_{i=0}^{d-1}\ket{ii}$ as the maximally entangled states of $\mathcal{H}_{AB}$ with $\mathrm{Dim}(\mathcal{H}_A)=\mathrm{Dim}(\mathcal{H}_B)=d.$

In this section, we will first recall the knowledge of the Schmidt number for a bipartite mixed state, then we will recall the definition and properties of SICs and MUBs, correspondingly. At last, we will present the definition of the distance to the set of states with Schmidt number less than k, $D_k(\cdot).$
\subsection{Schmidt Number}
Assume $\ket{\psi}_{AB}=\sum_{ij}c_{ij}\ket{ij}$ is a pure state in $\mathcal{H}_{AB}$ with $\mathrm{Dim}(\mathcal{H}_A)=d_A$ and $\mathrm{Dim}(\mathcal{H}_B)=d_B.$ 
There always exists orthonormal bases $\{\ket{\tilde{i}}_A\}$ and $\{\ket{\tilde{i}}_B\}$ in $\mathcal{H}_A$ and $\mathcal{H}_B,$ respectively such that $$\ket{\psi}_{AB}=\sum_{i=1}^k\sqrt{\lambda_i}\ket{\tilde{i}\tilde{i}},$$ here $\lambda_i>0$ and $\sum_i\lambda_i^2=1.$ Here the nonzero number $k$ is called the Schmidt number of $\ket{\psi}$ \cite{terhal2000schmidt}, $i.$ $e.$, $SN(\ket{\psi})=k.$ The Schmidt number of a mixed state $\rho_{AB}$ is defined as follows \cite{terhal2000schmidt},
\begin{align}
SN(\rho)=\min_{\rho=\sum_ip_i\ket{\psi_i}\bra{\psi_i}}\max_i SR(\ket{\psi_i}),
\end{align}
where the minimization takes over all the decompositions of $\rho_{AB}=\sum_ip_i\ket{\psi_i}\bra{\psi_i}.$ The Schmidt number is entanglement monotone, and it can be seen as a key quantity on the power of entanglement resources.\par 
Through the definition of Schmidt number, one can classify the states of $\mathcal{H}_{AB}$ as follows. Let 
\begin{align}
S_k=\{\rho|SN(\rho)\le k\},
\end{align}
due to the definition of $S_k$, we have $S_k\subset S_{k+1},$ and $S_1$ is the set of separable states.  Assume $\rho_{AB}\in S_k$ is a bipartite state in $\mathcal{H}_{AB}$ with $\mathrm{Dim}(\mathcal{H}_A)=\mathrm{Dim}(\mathcal{H}_B)=d$, the authors in \cite{terhal2000schmidt} showed that
\begin{align}
\tr(\rho\ket{\psi_d}\bra{\psi_d})\le \frac{k}{d},\label{f1}
\end{align}
 Besides, from the definition of $S_k$, we have $S_k$ is a convex set. Hence, we could construct the Schmidt number witness $W_k$ to validate a bipartite state $\rho_{AB}\in S_{k+1},$ due to the Hahn-Banach theorem, if
\begin{align*}
\tr(W_k\rho_k)\ge& 0 \hspace{5mm} \forall \rho_k\in S_{k},\\
\tr(W_k\rho)<& 0\hspace{5mm} \exists \rho\in\mathcal{D}(\mathcal{H}_{AB}),
\end{align*}
then we call $W_k$ a SN-(k+1) witness.

Then we recall the following result on the Schmidt number of a bipartite state obtained in \cite{terhal2000schmidt},
\begin{Lemma}\label{l1}
	Assume $\rho$ is a bipartite state on $\mathcal{H}_{AB}.$ $\rho$ has Schmidt number at least $k+1$ if and only if there exists a $k$-positive linear map $\Lambda_k$ such that 
	\begin{align}
	(I\otimes\Lambda_k)(\rho)\ngeq0.\label{kp}
	\end{align}
	The linear Hermitticity-preserving map $\Lambda$ is $k$-positive if and only if 
	\begin{align}
	(I\otimes\Lambda)(\ket{\psi_k}\bra{\psi_k})\ge0,
	\end{align}
	here $\ket{\psi_k}$ are arbitrary maximally entangled state with Schmidt number $k.$ Finally, if $\Lambda$ is $k$-positive, then $\Lambda^{\dagger}$ defined by $\tr A^{\dagger}\Lambda(B)=\tr\Lambda^{\dagger}(A^{\dagger})B$ for all $A$ and $B$, is also $k$-positive.
\end{Lemma}

In \cite{tomiyama1985geometry,terhal2000schmidt}, the authors showed a family of positive maps $\Lambda_p(X)$ with the following form
\begin{align}
\Lambda_p(X)=\tr(X)\mathbb{I}-pX,\label{klp}
\end{align}
where $X$ is a linear operator of $\mathcal{H}$. When $k\ge \frac{1}{p}>k+1$, $\Lambda_p$ is $k$-positive.
\subsection{SICs and MUBs}
Assume $\mathcal{H}_d$ is a Hilbert space with dimension $d,$ $\{{E_i}=\frac{1}{d}\ket{\phi_i}\bra{\phi_i}|i=1,2,\cdots,d^2\}$ is a positive operator valued measure(POVM) of $\mathcal{H}_d,$ here $\ket{\phi_i}$ are pure states with
\begin{align}
|\bra{\phi_j}\phi_k\rangle|^2=\frac{1}{d+1},\hspace{4mm}\forall j\ne k,
\end{align}
then $\{E_i\}_{i=1}^{d^2}$ is a SIC-POVM. The existence of SIC-POVMs in every dimension is still an open problem \cite{horodecki2022five}, readers who are interesting to the problem can refer to \cite{zauner1999quantum,scott2010symmetric,scott2017sics}. 

Assume $\rho$ is a state of $\mathcal{D}(\mathcal{H}_d),$ $\{E_i|i=1,2,\cdots,d^2\}$ is a SIC-POVM, then 
\begin{align}
\sum_{j=1}^{d^2}|\tr(P_j\rho)|^2=\frac{\tr\rho^2+1}{d+d^2},
\end{align}
which is showed in \cite{rastegin2014notes}.

Next we recall the definition of MUBs. Assume $\{\ket{e^l_i}|i=1,2,\cdots,d\}_{l=1}^L$ are $L$ orthonormal bases, and 
\begin{align*}
|\bra{e^m_j}f^{m^{'}}_k\rangle|^2=&\frac{1}{d}\hspace{5mm}\forall j,k\hspace{5mm} \textit{ m$\ne $$m^{'}$},\\
|\bra{e^m_j}f^{m}_k\rangle|^2=&\delta_{jk}.
\end{align*} 
then they are MUBs. For any space with dimension $d,$ there exists at most $d+1$ MUBs. If the upper bound is reached, the set of MUBs is called a complete set. It is well known that the complete sets of MUBs are existed when the dimension of the Hilbert space is a number with prime power. However, the existence of the complete sets of MUBs is unknown for arbitrary dimensional systems, even if the dimension is 6 \cite{horodecki2022five}.

Let $\{Q^{(\a)}_i=\ket{e^{(\a)}_i}\bra{e^{(\a)}_i}|i=1,2,\cdots,d\}_{\a=1}^L$ are $L$ MUBs and $\rho$ is a state in $\mathcal{D}(\mathcal{H}_d)$, then 
\begin{align}
	\sum_{\a=1}^L\sum_{i=1}^d|\tr(\r Q_i^{(\a)})|^2\le \tr(\r^2)+\frac{L-1}{d},
\end{align}
the above inequality is obtained in \cite{wu2009entropic}.

At last, we present the distance to the set $S_k$ for a bipartite state, $D_{k}(\cdot)$. Assume $\rho_{AB}$ is a bipartite mixed state, its distance to the set $S_k$ in terms of the Frobenius norm is defined as
\begin{align}
D_k(\rho)=\min_{\s\in S_k}\norm{\r-\s}_F,
\end{align}
where the minimum takes over all the state in $S_k.$
\section{k-positive maps based on SICs and MUBs}
Assume $\mathcal{H}_d$ is a Hilbert space with dimension $d,$ $\mathcal{M}=\{P_i=\frac{1}{d}\ket{\phi_i}\bra{\phi_i}\}_{i=1}^{d^2}$ is a SIC-POVMs. Next we present a class of $k$-positive maps $\Lambda(\cdot)$ with the help of the SIC-POVM $\mathcal{M}$. Let $\mathcal{O}$ be an orthogonal rotation in $\mathbb{R}^d$ around the axis $n_{*}=\frac{(1,1,\cdots,1)}{\sqrt{d}}$, that is, $\mathcal{O}n_{*}=n_{*},$
\begin{widetext}
\begin{align}
\Lambda(X)=  \frac{\mathbb{I}_d}{d}\tr(X)-h\sum_{g,l=1}^{d^2}\mathcal{O}_{gl}\tr[(X-\frac{\mathbb{I}_d}{d}\tr(X))P_l]P_g,\label{f2}
\end{align}
\end{widetext}
here $h=\sqrt{\frac{d^4+d^3}{(kd-1)(kd+k-2)}},$ $P_l$ and $P_k$ take over all the elements in a SIC-POVM $\mathcal{M}.$
\begin{Theorem}\label{t1}
	The map $\Lambda(\cdot)$ defined in (\ref{f2}) is $k$-positive.
\end{Theorem}
\begin{proof}
Due to the Lemma \ref{l1}, when we prove 
\begin{align}
(I\otimes \Lambda)(\ket{\psi_k}\bra{\psi_k})\ge 0,\label{f3}
\end{align}
for all maximally entangled state $\ket{\psi_k}$ with Schmidt number $k$, then we finish the proof. Here we utilize the following fact to prove (\ref{f3}): when $\rho$ is a Hermite matrix with trace 1, if $$\tr\rho^2\le \frac{1}{d-1},$$ then $\rho$ is a state \cite{be}.

Assume $\ket{\psi_k}=(U\otimes V)\sum_{i=0}^{k-1}\sqrt{\frac{1}{k}}\ket{ii}$, here $U$ and $V$ are arbitrary unitary operators of $\mathcal{H}_A$ and $\mathcal{H}_B$, respectively, then
\begin{align}
&\tr[(I\otimes \Lambda)(U\otimes V)(\ket{\psi_k}\bra{\psi_k})(U^{\dagger}\otimes V^{\dagger})]^2\nonumber\\
=&\tr\frac{1}{k^2}[(I\otimes\Lambda)(U\otimes V)(\sum_{i=0}^{k-1}\sum_{j=0}^{k-1}\ket{ii}\bra{jj})(U^{\dagger}\otimes V^{\dagger})]^2\nonumber\\
=&\tr\frac{1}{k^2}[\sum_{i,j=0}^{k-1}U\ket{i}\bra{j}U^{\dagger}\otimes\Lambda(V\ket{i}\bra{j}V^{\dagger})]^2\nonumber\\
=&\tr\frac{1}{k^2}[\sum_{i,j,m=0}^{k-1}U\ket{i}\bra{m}U^{\dagger}\otimes \Lambda(V\ket{i}\bra{j}V^{\dagger})\Lambda(V\ket{j}\bra{m}V^{\dagger})]\nonumber\\
=&\frac{1}{k^2}\tr\sum_{i,j=0}^{k-1}\Lambda(V\ket{i}\bra{j}V^{\dagger})\Lambda(V\ket{j}\bra{i}V^{\dagger})\label{f4}.
\end{align} 
Next we compute $\Lambda(\ket{i}\bra{j}),$ when $i=j$,
\begin{align}
\Lambda(V\ket{i}\bra{i}V^{\dagger})=\frac{\mathbb{I}_d}{d}-h\sum_{g,l=1}^{d^2}\mathcal{O}_{gl}\tr[(V\ket{i}\bra{i}V^{\dagger}-\frac{\mathbb{I}_d}{d})P_l]P_g,
\end{align}
when $i\ne j$,
\begin{align}
\Lambda(V\ket{i}\bra{j}V^{\dagger})=-h\sum_{g,l=1}^{d^2}\mathcal{O}_{gl}\tr[(V\ket{i}\bra{j}V^{\dagger})P_l]P_g.
\end{align}
\begin{widetext}
\begin{align}
(\ref{f4})=&\frac{1}{kd}+\frac{h^2}{k^2}\sum_{i=0}^{k-1}\sum_{g,l,s,t=1}^{d^2}\mathcal{O}_{gl}\mathcal{O}_{ts}\tr[(V\ket{i}\bra{i}V^{\dagger}-\frac{\mathbb{I}_d}{d})P_l]\tr[(V\ket{i}\bra{i}V^{\dagger}-\frac{\mathbb{I}_d}{d})P_s]\tr(P_gP_t)\nonumber\\
-&\frac{2h}{k^2}\sum_{i=0}^{k-1}\tr\sum_{g,l=1}^{d^2}\mathcal{O}_{gl}\tr[(V\ket{i}\bra{i}V^{\dagger}-\frac{\mathbb{I}_d}{d})P_l]\tr(P_g)+\frac{h^2}{k^2}\sum_{i\ne j}\sum_{g,l,s,t=1}^{d^2}\mathcal{O}_{gl}\mathcal{O}_{st}\tr[V\ket{i}\bra{j}V^{\dagger}P_l]\tr[V\ket{j}\bra{i}V^{\dagger}P_t]\tr(P_gP_s)\nonumber\\
=&\frac{1}{kd}+\frac{h^2}{d^2k^2}\sum_{i=0}^{k-1}\sum_{l=1}^{d^2}|\tr[(V\ket{i}\bra{i}V^{\dagger}-\frac{\mathbb{I}_d}{d})P_l]|^2+\frac{h^2}{k^2d^2}\sum_{i\ne j}\sum_{l=1}^{d^2}\tr(P_l V\ket{i}\bra{j}V^{\dagger})\tr(P_l V\ket{j}\bra{i}V^{\dagger})\nonumber\\
\le& \frac{1}{kd}+\frac{h^2(d-1)}{d^4(d+1)k}+\frac{h^2(k-1)}{d^4k}\nonumber\\
=&\frac{d^4+d^3+h^2kd+kh^2-2h^2}{kd^4(d+1)}\nonumber\\
=&\frac{1}{dk-1},
\end{align}
\end{widetext}
in the last inequality, we apply Lemma \ref{l2} in Sec. \ref{app} and the following derived in \cite{rastegin2014notes}
\begin{align}
\sum_{j=1}^{d^2}|\tr(P_j\rho)|^2=\frac{\tr\rho^2+1}{d+d^2},
\end{align}
hence we have $\Psi(\cdot)$ is $k$-positive.
\end{proof}

Based on Lemma \ref{l1}, we can provide a class of witnesses on detecting whether a bipartite state is in $S_k$ through the $k$-positive map defined in (\ref{f2}),
\begin{align}
W_k=\frac{h+d}{d^2}\mathbb{I}_d\otimes \mathbb{I}_d-h\sum_{gl}\mathcal{O}_{gl}\overline{P_l}\otimes P_g,\label{w1}
\end{align}
here $h=\sqrt{\frac{d^4+d^3}{(kd-1)(kd+k-2)}}.$

Next we present a class of $k$-positive maps $\Lambda(\cdot)$ based on the MUBs. Let $\{\ket{e^{\a}_i}|\a=1,2,\cdots,m\}_{i=1}^L$ be the MUBs, $\mathcal{O}^{(\a)}$ be $L$ orthogonal rotation in $\mathbb{R}^d$ around the axis $n_{*}=\frac{1}{d}(1,1,\cdots,1),$ that is, $\mathcal{O}^{(\a)}n_{*}=n_{*}.$ Based on the MUBs, we also present a set of $k$-positive maps,
\begin{widetext}
\begin{align}
\Theta_k(X)=  \frac{\mathbb{I}_d}{d}\tr(X)-h_s\sum_{\a=1}^L\sum_{g,l=1}^{d}\mathcal{O}^{(\a)}_{gl}\tr[(X-\frac{\mathbb{I}_d}{d}\tr(X))Q^{(\a)}_l]Q^{(\a)}_g,\label{f5}
\end{align}
\end{widetext}
here $h_s=\sqrt{\frac{1}{(dk-1)(Lk-L+d-1)}},$ $Q_l^{(\a)}=\ket{e_i^{\a}}\bra{e_i^{\a}}$, $\{\ket{e_i^{\a}}|\a=1,2,\cdots,L\}_{i=1}^d$ are $L$ MUBs. 
\begin{Theorem}
$\Theta_k(\cdot)$ defined in (\ref{f5}) are $k$-positive. 
\end{Theorem}
\begin{proof}
	Here we apply similar method of the proof of Theorem \ref{t1}. Assume $\ket{\psi_k}=(U\otimes V)\sum_{i=0}^{k-1}\sqrt{\frac{1}{k}}\ket{ii}$, here $U$ and $V$ are arbitrary unitary operators of $\mathcal{H}_A$ and $\mathcal{H}_B$, respectively, if $\tr[(I_k\otimes\Theta_k)(\ket{\psi_k}\bra{\psi_k})]^2\le \frac{1}{dk-1}$, then we finish the proof \cite{be}.
	\begin{align}
	&\tr[(I\otimes \Theta_k)(U\otimes V)(\ket{\psi_k}\bra{\psi_k})(U^{\dagger}\otimes V^{\dagger})]^2\nonumber\\
=&\frac{1}{k^2}\tr\sum_{i,j=0}^{k-1}\Theta_k(V\ket{i}\bra{j}V^{\dagger})\Theta_k(V\ket{j}\bra{i}V^{\dagger})\label{f6}
	\end{align}
	Next we compute $\Theta_k(V\ket{i}\bra{j}V^{\dagger}),$ when $i=j,$
	\begin{align}
	&\Theta_k(V\ket{i}\bra{i}V^{\dagger})\nonumber\\
	=&\frac{\mathbb{I}_d}{d}-h_s\sum_{\a=1}^L\sum_{g,l=1}^{d}\mathcal{O}^{(\a)}_{gl}\tr[(V\ket{i}\bra{i}V^{\dagger}-\frac{\mathbb{I}_d}{d})Q^{(\a)}_l]Q^{(\a)}_g,\label{f7}
	\end{align}
	when $i\ne j$,
	\begin{align}
	&\Theta_k(V\ket{i}\bra{j}V^{\dagger})\nonumber\\
	=&-h_s\sum_{\a=1}^L\sum_{g,l=1}^{d}\mathcal{O}^{(\a)}_{gl}\tr[(V\ket{i}\bra{j}V^{\dagger})Q^{(\a)}_l]Q^{(\a)}_g,\label{f8}
	\end{align}
	then based on (\ref{f7}) and (\ref{f8}), we have
	\begin{widetext}
	\begin{align*}
	(\ref{f6})\le &\frac{1}{k^2}[\frac{k}{d}+h_s^2\sum_{i=1}^k\sum_{\a=1}^L\sum_{n=1}^d|\tr(V\ket{i}\bra{i}V^{\dagger}-\frac{\mathbb{I}_d}{d})Q_n^{(\a)}|^2+h_s^2\frac{Lk(k-1)}{d}]\nonumber\\
	\le &\frac{1}{k^2}[\frac{k}{d}+kh_s^2(1-\frac{1}{d})+h_s^2\frac{Lk(k-1)}{d}]=\frac{1}{dk-1},
	\end{align*}
	In the inequalities, we have used Lemma \ref{l3} in Sec. \ref{app} the following derived in \cite{wu2009entropic}
	\begin{align}
	\sum_{\a=1}^L\sum_{l=1}^d|\tr(\r Q_l^{(\a)})|^2\le \tr(\r^2)+\frac{L-1}{d}
	\end{align}
	\end{widetext}
\end{proof}
Based on Lemma \ref{l1}, we can provide a class of witnesses on detecting whether a bipartite state is in $S_k$ through the $k$-positive map defined in (\ref{f5}),
\begin{align}
W_k=\frac{1+Lh_s}{d}\mathbb{I}_d\otimes \mathbb{I}_d-h_s\sum_{g,l=1}^d\sum_{\a=1}^L\mathcal{O}_{gl}^{(\a)}\overline{Q^{(\a)}_l}\otimes Q^{(\a)}_g,\label{w2}
\end{align}
here $h_s=\sqrt{\frac{1}{(dk-1)(Lk-L+d-1)}}.$
\begin{remark}
 When $k=1$, $h_s=\frac{1}{d-1},$ 
\begin{align}
W_k=\frac{d+L-1}{d(d-1)}\mathbb{I}_d\otimes \mathbb{I}_d-\frac{1}{d-1}\sum_{g,l=1}^d\sum_{\a=1}^L\mathcal{O}_{gl}^{(\a)}\overline{Q^{(\a)}_l}\otimes Q^{(\a)}_g,
\end{align}
which is the entanglement witnesses shown in \cite{chruscinski2018entanglement}. 
\end{remark}
\begin{remark}
	When $L=d+1,$ and $\mathcal{O}=\mathbb{I,}$ the $k$-positive map (\ref{f5}) can be written as
	\begin{align}
	W_k=\frac{\mathbb{I}_d}{d}{(1+h_c)}\tr(X)-h_cX,
	\end{align}
	here $h_c=\sqrt{\frac{1}{(dk-1)(kd+k-2)}}$, and 
	\begin{align*}
	\frac{1+h_c}{dh_c}=\frac{\frac{1}{h_c}+1}{d}=\frac{\sqrt{(dk-1)(kd+k-2)}+1}{d},
	\end{align*}
	as $\sqrt{(dk-1)(kd+k-2)}\in [k,k+1),$ we have this class of $k$-positive maps constructed from complete sets of MUBs are belong to the family of (\ref{klp}) in \cite{terhal2000schmidt}.
\end{remark}

\section{Applications}
\indent In this section, we will present the distance between the state and the set $S_k$ based on the witness obtained in the last section, the method here is based on \cite{shi2023lower}.

Assume $\rho$ is a bipartite mixed state, $Y_k$ is a Schmidt number $k$ witness of $\mathcal{H}_d\otimes\mathcal{H}_d$, let $a=\frac{\tr(Y_k)}{d^2}$, $b=\sqrt{\tr(Y_k^{\dagger}Y_k)-\frac{(\tr Y_k)^2}{d^2}},$ $V_k=\frac{Y_k-a\mathbb{I}\otimes\mathbb{I}}{b}$,
\begin{align}
D_{k}(\rho_{AB})=&\min_{\s\in S_{k}}\max_{\norm{W}_F=1}\tr(W(\rho_{AB}-\s_{AB}))\nonumber\\
\ge&|\tr[V_k(\rho-\omega)]\nonumber\\
=&|\tr[\frac{Y_k}{b}(\r-\o)-\frac{a}{b}(\r-\o)]\nonumber\\
=&|\tr[\frac{Y_k}{b}(\r-\o)]\nonumber\\
\ge&-\frac{1}{b}\tr Y_k\r,\label{wi}
\end{align}
in the first inequality, $\o$ is the optimal state in $S_k$, $V_k$ is a Hermite operator and $\norm{V_k}_F=1$. The last inequality is due to that $\o\in S_k,$ $\tr Y_k\o\ge0.$ 

Next we apply the witnesses obtained in (\ref{w2}) to show the distance between a bipartite state and the set $S_k$ relying on (\ref{wi}).
\begin{Theorem}
	Assume $\rho_{AB}$ is a bipartite mixed state on $\mathcal{H}_d\otimes\mathcal{H}_d$, let $\{\ket{e^{\a}_i}|\a=1,2,\cdots,m\}_{i=1}^L$ be the MUBs of the system $\mathcal{H}_d.$
Let $W_k$ be defined in (\ref{w2}), then
\begin{align}
D_k(\r)\ge -\frac{1}{\sqrt{h_s^2(Ld-L)}}\tr W_k\r
\end{align}
here $h_s=\sqrt{\frac{1}{(dk-1)(Lk-L+d-1)}}.$
\end{Theorem}
\begin{proof}
	Due to the (\ref{wi}), to obtain the lower bound of $D_k(\rho)$, we only need to compute the value of $b$. Let $Z_k=\sum_{g,l=1}^d\sum_{\a=1}^L\mathcal{O}_{gl}^{(\a)}\overline{Q^{(\a)}_l}\otimes Q^{(\a)}_g,$ 
	\begin{align}
&\tr Z_k\nonumber\\
=&\sum_{g,l=1}^d\sum_{\a=1}^L\mathcal{O}_{gl}^{(\a)}\tr\overline{Q^{(\a)}_l}\otimes Q^{(\a)}_g,\nonumber\\
=&\sum_{g,l=1}^d\sum_{\a=1}^L\mathcal{O}_{gl}^{(\a)}\nonumber\\
=&Ld,\label{y1}
\\
&\tr Z_k^{\dagger}Z_k\nonumber\\
=&\sum_{g,l,m,n=1}^d\sum_{\a,\b=1}^L\mathcal{O}_{gl}^{(\a)}\mathcal{O}_{mn}^{(\b)}\tr\overline{Q^{(\a)}_l}\overline{Q^{(\b)}_n}\otimes Q^{(\a)}_gQ^{(\b)}_m\nonumber\\
=&Ld+L^2-L,
\end{align}
then 
\begin{align}
b=&\sqrt{\tr(W_k^{\dagger}W_k)-\frac{(\tr W_k)^2}{d^2}}\nonumber\\
=&\sqrt{h_s^2(Ld-L)},
\end{align}
hence
\begin{align}
D_k(\r_{AB})\ge -\frac{1}{\sqrt{h_s^2(Ld-L)}}\tr W_k\r.
\end{align}
\end{proof}
\section{Conclusion}
Here we have presented families of $k$-positive maps in arbitrary dimensional systems based on the SIC POVMs and MUBs. Based on the $k$-positive maps, we also built the SN witnesses correspondingly. Then we compared the $k$-positive maps built from MUBs and the existing results. When $k=1,$ we found that the conclusion degrades into the map obtained in \cite{chruscinski2018entanglement}. When the $k$-positive maps built from a complete set of MUBs, they belong to the family obtained in \cite{terhal2000schmidt}. At last, we presented a defition of the distance between a bipartite state and the set $S_k.$ Moreover, we showed a lower bound of the distance based on the SN witnesses constructed from MUBs. Due to the important roles that higher dimensional systems played, our results can provide a reference for future work on the study of entanglement theory.

  \section{Acknowledgement}
X. S. was supported by the National Natural Science Foundation of China (Grant No. 12301580) and the Funds of College of Information Science and Technology, Beijing University of Chemical Technology (Grant No. 0104/11170044115).
\bibliographystyle{IEEEtran}
\bibliography{ref}

\section{Appendix}\label{app}
\begin{Lemma}\label{l2}
	Let $\mathcal{M}=\{M_j=\frac{1}{d}\ket{\phi_j}\bra{\phi_j}\}$ be a SIC-POVM in a $d$ dimensional system $\mathcal{H}$, and $\{\ket{i}|i=1,2,\cdots,d\}$ is any orthonormal base of $\mathcal{H}$, when $i\ne j,$
	\begin{align*}
	\sum_{j=1}^{d^2}\tr(M_j\ket{i}\bra{j})=\frac{1}{d^2}
	\end{align*}
\end{Lemma}
\begin{proof}
	Here we apply a similar method of the proof of Proposition 5 in \cite{rastegin2014notes}. Let 
	\begin{align}
	\ket{\Psi}=&\frac{1}{d^{3/2}}\sum_{j=1}^{d^2}\ket{\phi_j}\ket{\phi_j^{*}},\label{b1}
\\	\ket{\Phi_k}=&\frac{\sqrt{d+1}}{d^{3/2}}\sum_{j=1}^{d^2}\omega^{k(j-1)}\ket{\phi_j}\ket{\phi_j^{*}},\label{b2}
	\end{align}
	here $k=1,2,\cdots,d^2-1,$ and $\ket{\phi^{*}}$ is a vector such that its components are conjugate to the corresponding components of $\ket{\phi}$. The $d^2$ vectors (\ref{b1}),(\ref{b2})  construct an orthonormal basis of the space $\mathcal{H}\otimes\mathcal{H},$ hence,
	\begin{align}
	\ket{i}\bra{j}\otimes \mathbb{I}\ket{\Psi}=\sum_{k=1}^{d^2-1}a_k\ket{\Phi_k},\label{f9}
	\end{align}
	then
	\begin{align}
	a_k=&\bra{\Phi_k}\ket{i}\bra{j}\otimes \mathbb{I}\ket{\Psi}\nonumber\\
	=&\frac{\sqrt{d+1}}{d^3}\sum_{i,j=1}^{d^2}\omega^{-q(i-1)}\bra{\phi_i}\ket{i}\bra{j}\ket{\phi_j}\bra{\phi_j}\ket{\phi_i}\nonumber\\
	=&\frac{\sqrt{d+1}}{d}\sum_{i=1}^{d^2}\omega^{-q(i-1)}p_i,
	\end{align}
	here $p_i=\frac{\bra{\phi_i}\ket{i}\bra{j}\ket{\phi_i}}{d}.$ Next
	\begin{align}
	\bra{\Psi}(\ket{i}\bra{j}\otimes\mathbb{I})^2\ket{\Psi}=\frac{1}{d}\tr(\ket{i}\bra{j}i\rangle\bra{j})=0,
	\end{align}
	Through (\ref{f9}), 
	\begin{align}
	\sum_{k=1}^{d^2-1}a_k^{*}a_k=(d+1)\sum_{j=1}^{d^2}p_j^2-\frac{d+1}{d^2},
	\end{align}
	that is,
	\begin{align*}
\sum_{j=1}^{d^2}p_j^2=\frac{1}{d^2}
	\end{align*}
\end{proof}
\begin{Lemma}\label{l3}
	 Let $\{\ket{e^{\a}_i}|i=1,2,\cdots,m\}_{\a=1}^L$ be the MUBs in a $d$ dimensional system $\mathcal{H}$, and $\mathcal{N}^{(\a)}=\{Q_i^{(\a)}=\ket{e^{\a}_i}\bra{e^{\a}_i}|i=1,2,\cdots,m\}_{\a=1}^L$. Assume $\{\ket{i}|i=1,2,\cdots,d\}$ is any orthonormal base of $\mathcal{H}$, when $i\ne j,$
	\begin{align*}
	\sum_{\a}^L\sum_{i=1}^{d^2}\tr(Q_i^{(\a)}\ket{i}\bra{j})\le\frac{L}{d}
	\end{align*}
\end{Lemma}
The proof of Lemma \ref{l3} is similar to that of Lemma \ref{l2}, here we omit it. 
\end{document}